\documentclass[11pt,reqno]{amsart}
\usepackage{amsmath,latexsym,amssymb,amsfonts}
\usepackage{amscd,amssymb,epsfig}
\usepackage{graphicx}
\usepackage[pagebackref, colorlinks = true, linkcolor = blue, urlcolor  = blue, citecolor = red]{hyperref}

\usepackage[margin=1in]{geometry}

\newtheorem{theorem}{Theorem}[section]
\newtheorem{definition}[theorem]{Definition}
\newtheorem{proposition}[theorem]{Proposition}
\newtheorem{corollary}[theorem]{Corollary}
\newtheorem{lemma}[theorem]{Lemma}
\newtheorem{remark}[theorem]{Remark}

\newtheorem{question}[theorem]{Question}
\newtheorem{conjecture}[theorem]{Conjecture}

\def\id{\mathrm{id}}
\begin{document}

\title{On some classes of bipartite unitary operators}
\date{\today}

\author{Julien Deschamps}
\address{Dipartimento di Matematica, Universit\`a degli Studi di Genova, Via Dodecaneso, 35, 16146 Genova, Italy}
\email{deschamps@dima.unige.it}

\author{Ion Nechita}
\address{Zentrum Mathematik, M5, Technische Universit\"at M\"unchen, Boltzmannstrasse 3, 85748 Garching, Germany, and CNRS, Laboratoire de Physique Th\'eorique, Toulouse, France}
\email{nechita@irsamc.ups-tlse.fr}

\author{Cl\'ement Pellegrini}
\address{Institut de Math\'ematiques de Toulouse, Laboratoire de Statistique et de Probabilit\'e, Universit\'e Paul Sabatier (Toulouse III), 31062 Toulouse Cedex 9, France}
\email{clement.pellegrini@math.univ-toulouse.fr}

\subjclass[2000]{}

\begin{abstract}
We investigate unitary operators acting on a tensor product space, with the property that the quantum channels they generate, via the Stinespring dilation theorem, are of a particular type, independently of the state of the ancilla system in the Stinespring relation. The types of quantum channels we consider are those of interest in quantum information theory: unitary conjugations, constant channels, unital channels, mixed unitary channels, PPT channels, and entanglement breaking channels. For some of the classes of bipartite unitary operators corresponding to the above types of channels, we provide explicit characterizations, necessary and/or sufficient conditions for membership, and we compute the dimension of the corresponding algebraic variety. Inclusions between these classes are considered, and we show that for small dimensions, many of these sets are identical. 
\end{abstract}

\maketitle

\tableofcontents

\section*{Introduction}

In this work, we study some families of unitary operators acting on a tensor product of two finite dimensional Hilbert spaces, having some special properties in relation to the Stinespring dilation theorem. This fundamental result in operator algebras \cite{sti} states that any linear, completely positive, trace preserving map $L$ acting on $\mathcal M_n(\mathbb C)$ can be written as 
\begin{equation}\label{eq:Stinespring-intro}
L(\rho) = [\mathrm{id} \otimes \mathrm{Tr}](U(\rho \otimes \beta)U^*),
\end{equation}
where $U$ is a unitary operator acting on the tensor product $\mathbb C^n \otimes \mathbb C^k$, $\beta \in \mathcal M_k(\mathbb C)$ is a positive semidefinite matrix of unit trace, and $k$ is a large enough parameter ($k=n^2$ suffices). In quantum information theory \cite{nch}, the map $L$ is called a quantum channel, and the matrix $\beta$ is called a density matrix (or simply a quantum state). The Hilbert space $\mathbb C^k$ by which the original space $\mathbb C^n$ needs to be extended is called the environment, or the ancilla space. 

The starting point of our investigation is the remark that the channel $L$ in \eqref{eq:Stinespring-intro} depends, a priori, on the quantum state $\beta$. In the practice of quantum theory, the environment space $\mathbb C^k$ is usually large (most of the times much larger than the system space $\mathbb C^n$), and thus it is inconvenient to describe the aforementioned dependence of $L$ on $\beta$. More precisely, we would like to characterize the unitary operators $U$, for which, independently on the value of $\beta$, the channel $L$ given by \eqref{eq:Stinespring-intro} belongs to some given class $\mathcal L$ of quantum channels. 

In this work, we answer the question above for several classes $\mathcal L$ of relevance in quantum information theory: unitary conjugations $V \cdot V^*$, constant channels ($L(\rho)$ does not depend on the quantum state $\rho$), unital channels, ($L(I)=I$), mixed unitary channels (convex combinations of unitary conjugations), PPT channels (channels for which the Choi matrix has a positive semidefinite partial transpose), and entanglement breaking channels (channel which, acting on one half of any entangled state, yield separable states). In particular within our work, we give a partial positive answer to the conjecture formulated in \cite{adp,adp1}. In \cite{adp,adp1}, it was conjectured that convex combination of unitary conjugations can be obtained only by a very specific class of bipartite unitary operators (the set $\mathcal U_{mixed}$ in the present work). We show that this conjecture is true, under some additional assumptions.

It is very important to state at this time that we are not concerned with classes of channels, but with classes of bipartite unitary operators. Although these classes are defined in terms of channels, we are interested in characterizing the ``interaction'' unitaries $U$ with the property that, for all ancilla states $\beta$, the channel $L$ given by \eqref{eq:Stinespring-intro} has some fixed set of properties. A similar question was studied in \cite{jpc}, where the authors characterize the unitary operators $U$ having the property that the only matrices $\beta$ which give quantum channels in \eqref{eq:Stinespring-intro} are quantum states.

The paper is structured as follows. In Section \ref{sec:classes-def} we define the classes of unitary operators we are interested in, and we present some general properties. Sections \ref{sec:aut-const} - \ref{sec:block-diag} deal each with one or more of these classes. In Section \ref{sec:relations}, we collect some more relations between the different classes, and present some equality cases. We close the work with a section containing some open problems. Finally, in Appendix \ref{sec:block-svd}, we discuss the block singular value decomposition of operators. 

\bigskip

\noindent \emph{Acknowledgements.} We would like to thank St\'ephane Attal for suggesting the question, and Denis Bernard and Michael Kech for inspiring discussions. We are also grateful to Siddharth Karumanchi for pointing out to us reference \cite{kmwy} and sharing with us the unpublished notes \cite{kmwy2}. I.N.'s research has been supported by a von Humboldt fellowship and by the ANR projects {OSQPI} {2011 BS01 008 01} and  {RMTQIT}  {ANR-12-IS01-0001-01}. The three authors have been supported by the ANR project {StoQ} {ANR-14-CE25-0003-01}. 

\section{Some classes of unitary operators}
\label{sec:classes-def}

Let us fix, once and for all, the space on which the unitary operators we investigate will act. Put $\mathcal H_A = \mathbb C^n$, $\mathcal H_B = \mathbb C^k$ and let $\mathcal H_{AB} = \mathcal H_A \otimes \mathcal H_B =  \mathbb C^n \otimes \mathbb C^k$. The compact group of unitary operators acting on $\mathcal H_{AB} = \mathbb C^{nk}$ will be denoted by $\mathcal U_{nk}$ or, simply, $\mathcal U$:
\begin{equation*}
\mathcal U:= \mathcal U_{nk} = \{U \in \mathcal M_{nk}(\mathbb C) \, | \, UU^*=U^*U = I\}.
\end{equation*}

Starting from the trace linear form $\mathrm{Tr} : \mathcal M_k(\mathbb C) \to \mathbb C$ and the transposition operation $\mathrm{t}: \mathcal M_k(\mathbb C) \to \mathcal M_k(\mathbb C)$, define the \emph{partial trace} and the \emph{partial transposition} respectively by
\begin{align*}
\mathrm{Tr}_B &= \mathrm{id} \otimes \mathrm{Tr}\\
\mathrm{t}_B &= \mathrm{id} \otimes \mathrm{t}.
\end{align*}
In other words,
\begin{align*}
\mathrm{Tr}_B : \mathcal M_n(\mathbb C) \otimes \mathcal M_k(\mathbb C) &\to \mathcal M_n(\mathbb C)\\
A \otimes B &\mapsto (\mathrm{Tr} B)A
\end{align*}
and
\begin{align*}
\mathrm{t}_B : \mathcal M_n(\mathbb C) \otimes \mathcal M_k(\mathbb C) &\to \mathcal M_n(\mathbb C) \otimes \mathcal M_k(\mathbb C)\\
A \otimes B &\mapsto A \otimes B^\top.
\end{align*}
Note that we use the following notation for the transposition $\mathrm{t}(B)=B^\top$. For obvious aesthetic reasons, we shall write $X^\Gamma = \mathrm{t}_B(X)$.

We denote by $\mathcal M_k^{1,+}(\mathbb C)$ the set of $k$-dimensional density matrices (or quantum states)
$$\mathcal M_k^{1,+}(\mathbb C) := \{ \beta \in \mathcal M_k(\mathbb C) \, | \, \beta \geq 0 \text{ and } \mathrm{Tr}(\beta) = 1\}.$$

To a unitary transformation $U \in \mathcal U$ and a quantum state $\beta \in \mathcal M_k^{1,+}(\mathbb C)$, we associate the quantum channel $L_{U,\beta}:\mathcal M_n(\mathbb C) \to \mathcal M_n(\mathbb C)$ defined by the Stinespring formula (see, e.g., \cite[Section 8.2.2]{nch})
\begin{equation}\label{eq:Stinespring}
L_{U,\beta}(X) = \mathrm{Tr}_B(U\cdot X \otimes \beta \cdot U^*).
\end{equation}

Let us introduce some classes of quantum channels having the \emph{unitary invariance property}:
$$\forall \, L \in \mathcal L,\, V_1,V_2 \in \mathcal U_n, \qquad V_2L(V_1 \cdot V_1^*)V_2^* \in \mathcal L.$$

We present next a list of such classes of channels, leaving the task of verifying the unitary invariance property as an exercise for the reader:
\begin{itemize}
\item Unitary conjugations
$$
\mathcal L_{aut} = \{X \mapsto V X V^*\}_{V \in \mathcal U_n}.
$$
\item Constant channels
$$
\mathcal L_{const} = \{L : \mathcal M_n(\mathbb C) \to \mathcal M_n(\mathbb C) \, : \, \exists \, \sigma \in \mathcal M_n^{1,+}(\mathbb C) \text{ s.t. } \forall \, \rho \in \mathcal M_n^{1,+}(\mathbb C), \, L(\rho) = \sigma \}.
$$
\item Unital channels
$$
\mathcal L_{unital} = \{L : \mathcal M_n(\mathbb C) \to \mathcal M_n(\mathbb C) \, : \, L(I_n) = I_n \}.
$$
\item Mixed unitary channels
$$\mathcal L_{mixed} = \mathrm{conv} \{ X \mapsto VXV^*\}_{V \in \mathcal U_n}.$$
\item Positive partial transpose (PPT) channels
$$\mathcal L_{PPT} = \{L : \mathcal M_n(\mathbb C) \to \mathcal M_n(\mathbb C) \, : \,  \forall \, \rho \in \mathcal M_{n^2}^{1,+}(\mathbb C), \, \left[[L \otimes \mathrm{id}_n](\rho)\right]^\Gamma \geq 0\}.$$
\item Entanglement breaking channels
$$\mathcal L_{EB} = \{L : \mathcal M_n(\mathbb C) \to \mathcal M_n(\mathbb C) \, : \,  \forall \, \rho \in \mathcal M_{n^2}^{1,+}(\mathbb C), \, [L \otimes \mathrm{id}_n](\rho) \text{ is separable}\}.$$
\end{itemize}

We move next to the main definition of this paper, the classes of bipartite unitary channels we are interested in. These classes are defined in a natural way as the set of unitary operators inducing, via the Stinesping formula \eqref{eq:Stinespring}, independent of the state $\beta$ of the environment system $B$, quantum channels belonging to one of the classes above. This exact notion, in the case of degradable channels and entanglement breaking channels, has been considered respectively in \cite[Definition 15]{kmwy} and \cite{kmwy2}. More precisely, we define, for any $* \in \{aut,const, unital,mixed,PPT,EB\}$,
$$\mathcal U_{*} = \{U \in \mathcal U_{nk} \, | \, \forall \beta \in \mathcal M_k^{1,+}(\mathbb C), L_{U,\beta}\in \mathcal L_*\}.$$

 We have, in order:

\begin{align}
\label{eq:def-U-aut}
\mathcal U_{aut} &= \{U \in \mathcal U_{nk} \, | \, \forall \beta \in \mathcal M_k^{1,+}(\mathbb C), L_{U,\beta}(\rho) = V_\beta \rho V_\beta^* \text{ for some } V_\beta \in \mathcal U_n\}\\
\label{eq:def-U-const}
\mathcal U_{const} &= \{U \in \mathcal U_{nk} \, | \, \forall \beta \in \mathcal M_k^{1,+}(\mathbb C), L_{U,\beta} \text{ is a constant channel}\}\\
\label{eq:def-U-unital}
\mathcal U_{unital} &= \{U \in \mathcal U_{nk} \, | \, \forall \beta \in \mathcal M_k^{1,+}(\mathbb C), L_{U,\beta}(I) = I\}\\
\label{eq:def-U-mixed}
\mathcal U_{mixed} &= \{U \in \mathcal U_{nk} \, | \,  \forall \beta \in \mathcal M_k^{1,+}(\mathbb C), L_{U,\beta}(X) = \sum_{i=1}^r p_i(\beta) U_i(\beta) X U_i(\beta)^*\\
\notag
& \qquad \qquad \text{ with } p_i(\beta)\geq 0 , \sum_i p_i(\beta) = 1\text{ and } U_i(\beta) \in \mathcal U_n\}\\
\label{eq:def-U-PPT}
\mathcal U_{PPT} &= \{U \in \mathcal U_{nk} \, | \, \forall \beta \in \mathcal M_k^{1,+}(\mathbb C), L_{U,\beta} \text{ is a PPT channel}\}\\
\label{eq:def-U-EB}
\mathcal U_{EB} &= \{U \in \mathcal U_{nk} \, | \, \forall \beta \in \mathcal M_k^{1,+}(\mathbb C), L_{U,\beta} \text{ is an entanglement breaking channel}\}
 \end{align}
 
In relation to the class $\mathcal U_{aut}$, we also define (see Section \ref{sec:aut-const})
\begin{equation}\label{eq:def-U-single}
\mathcal U_{single} = \{U \in \mathcal U_{nk} \, | \, \text{the set }\{ L_{U,\beta}\, : \, \beta \in \mathcal M_k^{1,+}(\mathbb C)\} \text{ has 1 element}\}.
\end{equation}
 
 One of the original motivations of this work was to obtain a characterization of the set $\mathcal U_{mixed}$. As stepping stones towards a description of this set, we introduce the following classes of bipartite unitary operators:
\begin{align}
\label{eq:def-U-prob}
\mathcal U_{prob} &= \{U \in \mathcal U_{nk} \, | \, \exists U_i \in \mathcal U_n \text{ s.t. } \forall \beta \in \mathcal M_k^{1,+}(\mathbb C), L_{U,\beta}(X) = \sum_{i=1}^r p_i(\beta) U_i X U_i^*\\
\notag
& \qquad \qquad \text{ with } p_i(\beta)\geq 0 \text{ and } \sum_i p_i(\beta) = 1\}\\
\label{eq:def-U-prob-lin}
\mathcal U_{prob-lin} &= \{U \in \mathcal U_{nk} \, | \, \exists U_i \in \mathcal U_n \text{ s.t. } \forall \beta \in \mathcal M_k^{1,+}(\mathbb C), L_{U,\beta}(X) = \sum_{i=1}^r p_i(\beta) U_i X U_i^*\\
\notag
& \qquad \qquad \text{ with linear functions } p_i(\beta)\geq 0 \text{ and } \sum_i p_i(\beta) = 1\}\\
\label{eq:def-U-block-diag-A}
\mathcal U_{block-diag}^A &= \{U \in \mathcal U_{nk} \, | \, U = \sum_{i=1}^k U_i \otimes e_if_i^*, \\
\notag
&  \qquad  \qquad \text{ with } U_i \in \mathcal U_n \text{ and } \{ e_i \}, \{f_i\}  \text{ orthonormal bases in } \mathbb C^k\}\\
\label{eq:def-U-block-diag-B}
\mathcal U_{block-diag}^B &= \{U \in \mathcal U_{nk} \, | \, U = \sum_{i=1}^n e_if_i^* \otimes U_i, \\
\notag
&  \qquad  \qquad \text{ with } U_i \in \mathcal U_k \text{ and } \{ e_i \}, \{f_i\}  \text{ orthonormal bases in } \mathbb C^n\}
\end{align}

Let us mention at this point a simple but fundamental property of the sets of bipartite unitary matrices we have just introduced. 

\begin{lemma}
The \emph{local} unitary group $\mathcal U_n \times \mathcal U_k$ acts by left and right multiplication on $\mathcal U_{*}$, for all $* \in \{aut,const, unital, mixed, PPT, EB, single, prob, prob-lin, block-diag-A,block-diag-B\}$:
$$\forall \, U \in \mathcal U_{*}, \forall \, V_1,V_2 \in \mathcal U_n, \forall \, W_1, W_2 \in \mathcal U_k, \qquad (V_1 \otimes W_1) U (V_2 \otimes W_2) \in \mathcal U_{*}.$$
\end{lemma}
\begin{proof}
In the case of $* \in \{aut,const, unital, mixed, PPT, EB\}$, the proof follows from the bi-unitary invariance of the corresponding class $\mathcal L_*$ and from the fact that in the definition of $\mathcal U_*$, we require the condition to hold for all states on the environment $\beta \in \mathcal M_k^{1,+}(\mathbb C)$. The other cases are easy verifications, we leave the details to the reader.
\end{proof}

As suggested by the property above, the subgroup $\mathcal U_n \otimes \mathcal U_k \subseteq \mathcal U_{nk}$ plays an important role in our study. It turns out that the flip operator 
\begin{equation}\label{eq:flip}
F_n:\mathbb C^n \otimes \mathbb C^n \to \mathbb C^n \otimes \mathbb C^n, \qquad F_nx \otimes y = y \otimes x
\end{equation}
will also be of particular importance, in light of the following classical theorem.

\begin{theorem}(\cite[Theorem 3.1]{dli})
Let $G$ be a compact group such that $\mathcal U_n \otimes \mathcal U_k \subseteq G \subseteq \mathcal U_{nk}$. Then $G$ is one of the following
\begin{enumerate}
\item $G = \mathcal U_n \otimes \mathcal U_k$;
\item $G =  \mathcal U_{nk}$;
\item If $n=k$, $G = \langle \mathcal U_n \otimes \mathcal U_n, F_n \rangle$.
\end{enumerate}
\end{theorem}

 By direct computation, one can show that the following chain of inclusions holds :
$$\mathcal U^{A}_{block-diag} \subseteq \mathcal U^{A}_{prob} \subseteq \mathcal U^{A}_{mixed} \subseteq \mathcal U^{A}_{unital} \subseteq \mathcal U.
$$

Note that one can define ``$B$''-versions of the above sets, in an obvious way, by swapping the tensor factors $A$ and $B$ ( above inclusions are still true for $\mathcal U^B_*$).

One of the main focuses of the current work will be to understand which are the inclusions above which are strict and which are actually equalities. 

We end this section by showing that if two interaction unitary operators $U,V$ generate the same channels for all states $\beta$ on the environment, then they are related by a unitary operator acting on the environment $\mathbb C^k$; this result will turn out to be useful later on. 

\begin{lemma}\label{lem:same-channels}
Let $U,V \in \mathcal U_{nk}$ be two bipartite unitary operators. The following two assertions are equivalent:
\begin{enumerate}
\item \label{it:FUbeta-equals-FVbeta} For all density matrix on $\beta \in \mathcal M_k^{1,+}(\mathbb C)$, the quantum channels $U$ and $V$ induce are equal: $L_{U,\beta}(\rho) = L_{V,\beta}(\rho)$ for all $\rho \in \mathcal M_n^{1,+}(\mathbb C)$;
\item \label{it:U-equals-IW-times-V} There exists a unitary operator $W \in \mathcal U_k$ such that $U = (I \otimes W) V$.
\end{enumerate}
\end{lemma}
\begin{proof}
We only prove ``\ref{it:FUbeta-equals-FVbeta} $\implies$ \ref{it:U-equals-IW-times-V}'', since the converse follows from direct calculation. We start form the hypothesis, 
\begin{equation*}
\forall \rho \in \mathcal M_n^{1,+}(\mathbb C), \, \forall \beta \in \mathcal M_k^{1,+}(\mathbb C), \quad [\mathrm{id} \otimes \mathrm{Tr}](U \rho \otimes \beta U^*) =  [\mathrm{id} \otimes \mathrm{Tr}](V \rho \otimes \beta V^*).
\end{equation*}
By linearity, we can replace $\rho$, resp.~$\beta$, by arbitrary complex matrices $A \in \mathcal M_n(\mathbb C)$, resp.~$B \in \mathcal M_k(\mathbb C)$. Again using linearity, we replace the simple tensor $A \otimes B$ by a general element $X \in \mathcal M_n(\mathbb C) \otimes \mathcal M_k(\mathbb C)$ to obtain the first line below, and we continue by obvious equivalent reformulations:
\begin{align*}
\forall X \in \mathcal M_{nk}(\mathbb C), \quad &[\mathrm{id} \otimes \mathrm{Tr}](U X U^*)  = [\mathrm{id} \otimes \mathrm{Tr}](V X V^*) \\
\forall X \in \mathcal M_{nk}(\mathbb C), \,  \forall A \in \mathcal M_{n}(\mathbb C), \quad & \mathrm{Tr}(U X U^* A \otimes I)  = \mathrm{Tr}(V X V^* A \otimes I) \\
\forall X \in \mathcal M_{nk}(\mathbb C), \,  \forall A \in \mathcal M_{n}(\mathbb C), \quad & \mathrm{Tr}( X U^* A \otimes IU)  = \mathrm{Tr}( X V^* A \otimes IV) \\
\forall A \in \mathcal M_{n}(\mathbb C), \quad &U^* A \otimes I U  = V^* A \otimes I V \\
\forall A \in \mathcal M_{n}(\mathbb C), \quad &A \otimes I (UV^*)  = (UV^*) A \otimes I  \\
\exists W \in \mathcal M_{k}(\mathbb C), \quad &UV^* = I \otimes W
\end{align*}
and we conclude the proof by noticing that $W$ has to be unitary, since $UV^*$ is.
\end{proof}

\section{Bipartite unitary operators producing unitary conjugations and constant channels}
\label{sec:aut-const}

In this section we provide characterizations of the sets $\mathcal U_{aut}$ \eqref{eq:def-U-aut},  $\mathcal U_{single}$ \eqref{eq:def-U-single}, and $\mathcal U_{const}$ \eqref{eq:def-U-const}, showing that only tensor products of unitary operators (resp.~ flipped tensor products) belong to these classes. 

We start with the set of automorphisms, that is the set of unitary conjugation channels. 

\begin{theorem}\label{thm:aut}
Let $U \in \mathcal U_{nk}$ a bipartite unitary operator such that, for all quantum states $\beta \in \mathcal M_k^{1,+}(\mathbb C)$, there exists an unitary operator $V_\beta \in \mathcal U_n$ such that $L_{U, \beta}(\rho) = V_\beta \rho V_\beta^*$ for all $\rho \in \mathcal M_n^{1,+}(\mathbb C)$. Then, there exist unitary operators $V \in \mathcal U_n$ and $W \in \mathcal U_k$ such that $U = V \otimes W$. In other words,
$$\mathcal U_{aut} = \{V \otimes W \, : \, V \in \mathcal U_n, \, W \in  \mathcal U_k\}.$$
\end{theorem}
\begin{proof}
The proof is easy, and consists of two steps: we show first that the unitary operators $V_\beta$ can be chosen to not depend on $\beta$, and then, using Lemma \ref{lem:same-channels}, we show that the unitary $U$ has the required form. 

Let us first introduce some notation. To any matrix 
$$\mathcal M_n(\mathbb C) \ni A = \sum_{i,j=1}^n A_{ij} e_ie_j^*,$$
 associate its \emph{vectorization} $a = \mathrm{vect}(A)$ defined by
 $$\mathbb C^n \otimes \mathbb C^n \ni a = \sum_{i,j=1}^n A_{ij} e_i \otimes e_j,$$
 where $\{e_i\}_{i=1}^n$ is some fixed orthonormal basis of $\mathbb C^n$. 

Denote $V:=V_{I/k}$, the unitary which appears for $\beta = I/k$, and let $v = \mathrm{vect}(V)$. For an arbitrary $\beta \in \mathcal M_k^{1,+}(\mathbb C)$, let $\tilde \beta = (\beta + I/k)/2$. Since quantum channels are linear in $\beta$, we have that
$$\forall \rho \in \mathcal M_n^{1,+}(\mathbb C), \qquad 2V_\beta \rho V_\beta^* = V_{\tilde \beta} \rho V_{\tilde \beta}^* + V \rho V^*.$$
At the level of Choi matrices, the above equation reads 
$$2 v_\beta v_\beta^* = v_{\tilde \beta} v_{\tilde \beta}^* + v v^*.$$
On the left side of the equation above the operator is a rank one operator. This way, in order that the sum in the right side is a rank one operator, the vector has to be proportional. Since all the involved vectors are of norm $\sqrt n$, it follows that they should all be the same, up to a phase. The same holds for the unitary operators, which concludes the first step of the proof. 

Let $\tilde U := V \otimes I$. It is easy to see that
$$\forall \rho \in \mathcal M_n^{1,+}(\mathbb C), \, \forall \beta \in \mathcal M_k^{1,+}(\mathbb C), \qquad L_{\tilde U, \beta}(\rho) = V \rho V^* = L_{U, \beta}(\rho),$$
so, by Lemma \ref{lem:same-channels}, there exists an unitary operator $W \in \mathcal U_k$ such that $U = (I \otimes W) \tilde U$, and the proof is complete. 
\end{proof}

We show next that the class $\mathcal U_{single}$, i.e.~ the set of unitary operators $U \in \mathcal U_{nk}$ with the property that the map $\beta \mapsto L_{U,\beta}$ is constant, is actually identical to $\mathcal U_{aut}$. 

For the proof of Theorem \ref{thm:single}, we need the following lemma.

\begin{lemma}\label{lem:orthogonal-zero-trace}
Let $M \in \mathcal M_n(\mathbb C) \otimes \mathcal M_k(\mathbb C)$ be a given matrix. The following conditions are equivalent:
\begin{enumerate}
\item For all matrices $X \in \mathcal M_n(\mathbb C)$ and $Y \in \mathcal M_k(\mathbb C)$ such that $\mathrm{Tr} Y = 0$, we have $\mathrm{Tr}(M \cdot X \otimes Y) = 0$;
\item There exist a matrix $A \in \mathcal M_n(\mathbb C)$ such that $M = A \otimes I_k$.
\end{enumerate}
\end{lemma}
\begin{proof}
The non-trivial implication follows from the following equation
$$[\mathcal M_n(\mathbb C) \otimes (\mathcal M_k(\mathbb C) \ominus \mathbb C I_k)]^\perp = \mathcal M_n(\mathbb C) \otimes \mathbb C I_k.$$
\end{proof}

\begin{theorem}\label{thm:single}
Let $U \in \mathcal U_{nk}$ a bipartite unitary operator such that, for all quantum states $\rho \in \mathcal M_n^{1,+}(\mathbb C)$ and $\beta,\gamma \in \mathcal M_k^{1,+}(\mathbb C)$, we have that $L_{U, \beta}(\rho) = L_{U, \gamma}(\rho)$. Then, there exist unitary operators $V \in \mathcal U_n$ and $W \in \mathcal U_k$ such that $U = V \otimes W$. In other words, the set $\mathcal U_{single}$ defined in \eqref{eq:def-U-single} is equal to
$$\mathcal U_{single} = \{V \otimes W \, : \, V \in \mathcal U_n, \, W \in  \mathcal U_k\} = \mathcal U_{aut}.$$
\end{theorem}

\begin{proof}
By using linearity, the hypothesis translates to the following equality
$$\forall X,Y \in \mathcal M_n(\mathbb C),\, \forall Z \in \mathcal M_k(\mathbb C) \text{ s.t. } \mathrm{Tr} Z = 0, \quad \mathrm{Tr}[U(X \otimes Z) U^* (Y \otimes I_k)] = 0.$$
By reshaping the operator $U$, the previous equality can be re-written as 
\begin{equation}\label{eq:hat-U}
\mathrm{Tr}[Y^\top \otimes X \otimes Z \cdot \hat U^* \hat U] = 0.
\end{equation}
where $\hat U \in \mathcal M_{k \times n^2k}(\mathbb C)$ is the operator
$$\hat U = \sum_{i,j=1}^n \sum_{x,y = 1}^k \langle e_i \otimes f_x, U e_j \otimes f_y \rangle f_x\cdot (e_i^* \otimes e_j^* \otimes f_y^*),$$
see Figure \ref{fig:hat-U} for a graphical representation of $\hat U$ and of the previous equality (for the Penrose graphical notation for tensors, in the form used here, see \cite[Section 3]{cne10a}). From Lemma \ref{lem:orthogonal-zero-trace}, it follows that there exist an operator $A \in \mathcal M_{n^2}(\mathbb C)$ such that $\hat U^* \hat U = A \otimes I_k$. Since the rank of $\hat U^* \hat U$ is at most $k$, it follows that $A$ has rank at most 1, i.e.~ there exist a vector $a \in \mathbb C^n \otimes \mathbb C^n$ such that $A = aa^*$. From the equality $\hat U^* \hat U = aa^* \otimes I_k$, we deduce that there exists an operator $W \in \mathcal U_k$ such that $\hat U = a^* \otimes W$. Writing
$$\mathcal M_n(\mathbb C) \ni V = \sum_{i,j=1}^n \langle e_i \otimes e_j,a\rangle e_ie_j^*,$$
we get $U = V \otimes W$. The fact that $V \in \mathcal U_n$ follows from the unitarity of $U$, and this concludes the proof of the first implication. 
The fact that tensor product unitary operators belong to $\mathcal U_{single}$ can be verified by direct computation.
\end{proof}
\begin{figure}[htbp] 
\includegraphics{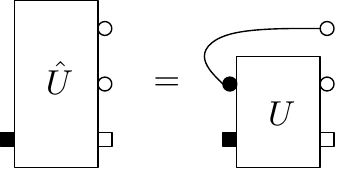} \qquad \includegraphics{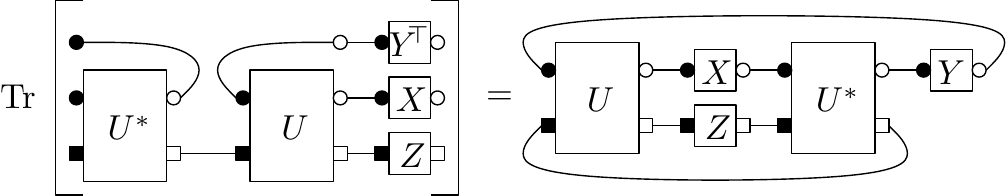}
\caption{The reshaping $\hat U$ of a unitary operator $U$ (left) and the trace inequality (right).} 
\label{fig:hat-U}
\end{figure}

For the case of $\mathcal U_{const}$, it is easy to see that it is related to the previous case via a flip operation, although there is a slight technical complication. This question has also been considered in \cite{kmwy2}.

\begin{theorem}\label{thm:const}
Consider the set $\mathcal U_{const}$ of bipartite unitary operators such that, for all quantum states $\rho,\sigma \in \mathcal M_n^{1,+}(\mathbb C)$ and $\beta \in \mathcal M_k^{1,+}(\mathbb C)$, we have that $L_{U, \beta}(\rho) = L_{U, \beta}(\sigma)$. 

If $k \neq rn$ for $r = 1,2, \ldots$, then $\mathcal U_{const}$ is empty.

 If $k=rn$ for some positive $r$, then we have
\begin{equation}\label{eq:U-const}
\mathcal U_{const} = \{(I_n \otimes V)(F_n \otimes I_r)(I_n \otimes W) \, : \,  V,W \in \mathcal U_k\},
\end{equation}
where $F_n \in \mathcal U_{n^2}$ denotes the flip operator \eqref{eq:flip}; see Figure \ref{fig:U-const} for the diagrammatic representation of such an operator.
\end{theorem}

\begin{figure}[htbp] 
\includegraphics{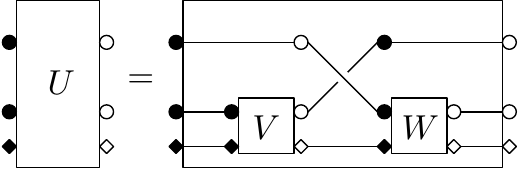} 
\caption{The diagrammatic representation of a general element from $\mathcal U_{const}$, in the case where $k=nr$. The diamond shaped labels correspond to the vector space $\mathbb C^r$.} 
\label{fig:U-const}
\end{figure}

\begin{proof}
Let us start with the easy implication, considering an operator $U$ as in \eqref{eq:U-const}. By direct computation, one can see that
$$L_{U,\beta}(\rho) = [\mathrm{id}_n \otimes \mathrm{Tr}_r](\beta),$$
proving that, for arbitrary $\beta$, the channel $L_{U,\beta}$ is constant.

The proof of the difficult implication starts in the same way as the one of Theorem \ref{thm:single}. The hypothesis that the channels $L_{U,\beta}$ are constant translates to the following condition: 
$$\forall X,Y \in \mathcal M_n(\mathbb C),\, \forall Z \in \mathcal M_k(\mathbb C) \text{ s.t. } \mathrm{Tr} X = 0, \quad \mathrm{Tr}[U(X \otimes Z) U^* (Y \otimes I_k)] = 0.$$
As before, after reshaping the matrix $U$ into $\hat U$, the previous relation becomes \eqref{eq:hat-U}, with the difference that this time, we have $\mathrm{Tr} X = 0$. Using Lemma \ref{lem:orthogonal-zero-trace}, we conclude that there exists an operator $A \in \mathcal M_{nk}(\mathbb C)$ such that 
\begin{equation}\label{eq:hat-U-star-hat-U}
\hat U^* \hat U = I_n^{(2)} \otimes A^{(1,3)},
\end{equation}
where the superscripts indicate on which factor of the tensor products the operators are acting ($A$ acts on the first copy of $\mathbb C^n$ and on $\mathbb C^k$, while the identity operator acts on the second copy of $\mathbb C^n$). Using the fact that $U$ is unitary, we have that the rank of the matrix $\hat U^* \hat U$ is \emph{precisely} $k$, so we must have $k = n\times \mathrm{rank}( A)$. We conclude that if $k$ is not a multiple of $n$, then operators $U$ with the above property cannot exist. We put now $r:=\mathrm{rank}(A)$, so that $k=rn$. Since $\hat U^* \hat U$ is a positive semidefinite matrix, $A$ is a positive semidefinite matrix of rank $r$, so we write $A = B^*B$, for a matrix $B \in \mathcal M_{r \times nk}(\mathbb C) = \mathcal M_{r \times n^2r}(\mathbb C)$. The fact that $U$ is unitary translates to the following equality
\begin{equation}\label{eq:hat-B}
\hat B \hat B^* = I_{nr},
\end{equation}
where 
$$\mathcal M_{nr}(\mathbb C) \ni \hat B = \sum_{i,j=1}^n \sum_{x,y=1}^r \langle g_x, B e_i \otimes e_j \otimes g_y \rangle e_i \otimes g_x \cdot e_j^* \otimes g_y^*,$$
for some orthonormal basis $\{g_1, \ldots g_r\}$ of $\mathbb C^r$; see Figure \ref{fig:hat-B} for a graphical representation of the previous equalities. Thus $W:=\hat B$ is a unitary operator. From equation \eqref{eq:hat-U-star-hat-U}, we find that there is another unitary operaor $V \in \mathcal U_{nr}$ such that 
$$U^{(1,2,3)} = \left(I_n^{(1)} \otimes V^{(2,3)}\right) \cdot \left(F_n^{(1,2)} \otimes I_r^{(3)}\right) \cdot \left(I_n^{(1)} \otimes W^{(2,3)}\right),$$
where the three tensor legs correspond to $\mathbb C^n \otimes \mathbb C^k \cong \mathbb C^n \otimes \mathbb C^n \otimes \mathbb C^r$. The conclusion follows. 
\begin{figure}[htbp] 
\includegraphics{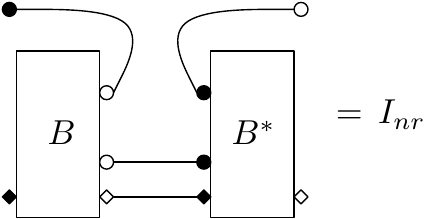} 
\caption{A diagramatic representation of \eqref{eq:hat-B}, from which the unitarity of $\hat B$ follows.} 
\label{fig:hat-B}
\end{figure}
\end{proof}

\begin{remark}
If $n=k$ in the previous result, equation \eqref{eq:U-const} can be written as $U = F_n (V \otimes W)$ for a pair of unitary operators $V,W \in \mathcal U_n$, so in this case we have
$$\mathcal U_{const} = F_n \mathcal U_{single} = \mathcal U_{single}F_n.$$
\end{remark}

\section{Bipartite unitary operators producing unital channels}

In this section we study the set $\mathcal U^{A}_{unital}$ of bipartite unitary operations which yield unital channels for every choice of the state on the auxiliary space. 

Using linearity, one can extend the definition \eqref{eq:def-U-unital} to the whole space of $k \times k$ complex matrices:
$$\mathcal U_{unital} =  \{U \in \mathcal U \, | \, \forall B \in \mathcal M_k(\mathbb C), L_{U,B}(I) =(\mathrm{Tr} B) I\}.$$

\begin{theorem}\label{unital}
One has 
$$\mathcal U_{unital} = \mathcal U_{nk} \cap \mathcal U_{nk}^\Gamma.$$
In particular, $\mathcal U^{B}_{unital} = \mathcal U^{A}_{unital}$.
\end{theorem}
\begin{proof}
Let us first show the ``$\subseteq$'' inclusion in the above equality. Take $U \in \mathcal U_{unital}$ and put $V = U^\Gamma \in \mathcal M_{nk}(\mathbb C)$, and $W=VV^*$. One has, for any $B \in \mathcal M_k(\mathbb C)$,
\begin{align}
\label{eq:UGamma-first}(\mathrm{Tr} B)I_n = L_{U,B}(I_n) &= [\mathrm{id} \otimes \mathrm{Tr}] (U \cdot I\otimes B \cdot U^*) \\
\label{eq:UGamma-last}&=  [\mathrm{id} \otimes \mathrm{Tr}] (VV^* \cdot I \otimes B^\top)\\
\nonumber&=  [\mathrm{id} \otimes \mathrm{Tr}] (W \cdot I \otimes B^\top)
\end{align}
By block-decomposing $W$ in an arbitrary orthonormal basis $e_i$ of $\mathbb C^k$
$$W = \sum_{i,j=1}^k W_{ij} \otimes e_ie_j^*,$$
we get that for all $B \in \mathcal M_k(\mathbb C)$,
$$(\mathrm{Tr} B)I_n = \sum_{i,j=1}^k W_{ij} \cdot \langle e_j, Be_i \rangle.$$
Choosing $B = e_ie_j^*$, we get $W_{ij}=\delta_{ij}I_n$ and hence $W=I_{nk}$. In other words, $V=U^\Gamma \in \mathcal U_{nk}$, which finishes the proof of the first inclusion.

The second inclusion follows by working backwards the previous arguments: since $V=U^\Gamma \in \mathcal U_{nk}$, equations \eqref{eq:UGamma-last} and \eqref{eq:UGamma-first} hold.
\end{proof}

Since both sets $\mathcal U_{nk}$ and $\mathcal U_{nk}^\Gamma$ are algebraic varieties (i.e.~ they can be describes as the zero-set of a system of polynomial equations), we obtain the following corollary.
\begin{corollary}
The set $\mathcal U_{unital}$ is a real algebraic variety.
\end{corollary}

We are going to investigate next $\mathcal U_{unital}$, as an algebraic variety. We compute first the dimension of the ``enveloping tangent space''  of $\mathcal U_{unital}$ at a point which is a block-diagonal unitary
$$U = \sum_{i=1}^k U_i \otimes e_i f_i^* \in \mathcal U_{block-diag}^A.$$
The notion of enveloping tangent spaces was introduced in \cite{tzy} (also called defect), and it is simply defined by (see also \cite{ban})
$$\widetilde T_U(\mathcal U_{unital}) := T_U(\mathcal U_{nk}) \cap T_U(\mathcal U_{nk}^\Gamma).$$

\begin{proposition}
The dimension of the enveloping tangent space of $\mathcal U_{unital}$ at a point which is a block-diagonal unitary of the form
$$U = \sum_{i=1}^k U_i \otimes e_i f_i^*$$
is given by
\begin{equation}\label{eq:dimension-tangent-space-classical}
D_U = \sum_{i,j=1}^k \sum_{(\lambda_{x},d_x) \in \Lambda_{ij}} d_x^2,
\end{equation}
where $\Lambda_{ij}$ is the set $\{(\lambda_x, d_x)\}$ where $\lambda_x$ is an eigenvalue of the unitary operator $U_i U_j^*$ having multiplicity $d_x$.
\end{proposition}
\begin{proof}

A matrix $A \in \mathcal M_{nk}(\mathbb C)$ is an element of the enveloping tangent space at $U$ if and only if both matrices $U + \varepsilon A$ and $(U + \varepsilon A)^\Gamma$ are unitary, up to the first order in $\varepsilon$. The unitarity of $U + \varepsilon A$ is equivalent to the condition $UA^* + AU^* = 0$, while the unitarity of $(U + \varepsilon A)^\Gamma$ is equivalent to  $U(A^\Gamma)^* + A^\Gamma U^* = 0$ (note that we have used the fact $U = U^\Gamma$).

Writing $A$ as a block matrix 
$$A = \sum_{i,j} A_{ij} \otimes e_i f_j^*,$$
the first condition  $UA^* + AU^* = 0$ is equivalent to the following system of equations
$$\forall i,j, \qquad U_i A_{ji}^* + A_{ij} U_j^* = 0,$$
while the condition $U(A^\Gamma)^* + A^\Gamma U^* = 0$ is equivalent to 
$$\forall i,j, \qquad U_i A_{ij}^* + A_{ji} U_j^* = 0.$$

First, let us note that the diagonal blocks $A_{ii}$ appear only in two identical equations 
$$U_i A_{ii}^* + A_{ii}U_i^* = 0.$$
The general solution to the equation above is $A_{ii} = B_{ii} U_i$, where $B_{ii}$ is an arbitrary anti-hermitian matrix ($B_{ii} + B_{ii}^*=0$). Hence, the total dimension of the diagonal blocks of $A$ is $kn^2$. Note that this corresponds to the case $i=j$ in formula \eqref{eq:dimension-tangent-space-classical}: there, $\Lambda_{ii} = \{(1,n)\}$.

Let us now study off-diagonal blocks of $A$. Again, the equations are decoupled: for $i<j$, one has to solve
\begin{align}
\label{eq:system-Aij-first}U_i A_{ji}^* + A_{ij} U_j^* &= 0\\
\label{eq:system-Aij-second}U_i A_{ij}^* + A_{ji} U_j^* &= 0.
\end{align}
From the first equation, one finds $A_{ji} = -U_j A_{ij}^* U_i$. Plugging this into the second equation, we have to solve now
$$RA_{ij} - A_{ij}S = 0,$$
where $R =  U_i U_j^*$ and $S = U_j^* U_i$. This is the well-known \emph{Sylvester equation}. From the analysis in \cite[Chapter VIII]{gan}, the dimension of the solution space of this homogenous equation depends of the Jordan block structure of the matrices $R$ and $S$. Since in our case both $R$ and $S$ are unitary (hence diagonalizable), the Jordan blocks have unit dimension. Moreover, $R$ and $S$ have the same spectrum $\Lambda_{ij}$. It follows from \cite[Chapter VIII, eq.~ (19)]{gan} that the \emph{complex} dimension of the solutions of the system \eqref{eq:system-Aij-first}-\eqref{eq:system-Aij-second} is precisely 
$$\sum_{(\lambda_{x},d_x) \in \Lambda_{ij}} d_x^2,$$
and the proof is complete.
\end{proof}
The proof above can be adapted mutatis mutandis to the case of $(B)$-classical unitary operators, as follows.
\begin{corollary}
The same result holds for $B$-block-diagonal unitary operators of the from 
$$U = \sum_{i=1}^n e_i f_i^* \otimes U_i.$$
\end{corollary}
\begin{corollary}
The dimension of the enveloping tangent space of $\mathcal U_{unital}$ at a product unitary operator $U = V \otimes W$ is $n^2k^2$, which is also the dimension of $\mathcal U_{nk}$.
\end{corollary}
\begin{corollary}
For $k=2$, the dimension of the enveloping tangent space of $\mathcal U_{unital}$ at a point $U = I \oplus V =  I \otimes e_1 f_1^* + V \otimes e_2 f_2^*$ is 
$$D_{I \oplus V} = 2n^2 + 2 \sum_{\lambda} d_\lambda^2,$$
where $d_\lambda$ are the multiplicities of the eigenvalues $\lambda$ of $V$.
\end{corollary}

\begin{corollary}
Consider a block-diagonal unitary operator
$$U = \sum_{i=1}^k U_i \otimes e_i f_i^*,$$
where the operators $U_i$ are in \emph{generic position}:
$$\forall i \neq j, \qquad U_i U_j^* \text{ has a simple spectrum}.$$
The dimension of the enveloping tangent space of $\mathcal U_{unital}$ at $U$ is then
\begin{equation}\label{eq:dim-U-unital}
D_{U} = kn^2 + nk^2 - nk.
\end{equation}
Note that the expression above is symmetric in $n$ and $k$.
\end{corollary}

We conjecture that the expression \eqref{eq:dim-U-unital} is the dimension of $\mathcal U_{unital}$, as an algebraic variety, see Conjecture \ref{cnj:dim-U-unital}.

\section{Bipartite unitary operators producing PPT channels}
\label{sec:PPT}

We consider in this section $PPT$ channels and bipartite unitary operators which produce such channels via the Stinespring formula, independent of the state of the environment. 
 
Recall that the \emph{maximally entangled state} is the matrix (here, we drop the normalization constant)
$$\mathcal M_{n^2}(\mathbb C) \ni \Omega_n := \sum_{i,j=1}^n e_iej^* \otimes e_i ej^*.$$
A quantum channel $L$ is said to be \emph{PPT} if and only if its Choi matrix 
 $$C_L := [\mathrm{id} \otimes L](\Omega_n)$$
 is PPT, i.e.~ $C_L^\Gamma \geq 0$. Hence, the set $\mathcal U_{PPT}$ admits the following characterization:
 $$\mathcal U_{PPT} = \{U \in \mathcal U_{nk} \, : \, \text{the map } \beta \mapsto C_{L_{U,\beta}}^\Gamma \text{ is positive}\}.$$
 
 \begin{figure}[htbp] 
\includegraphics{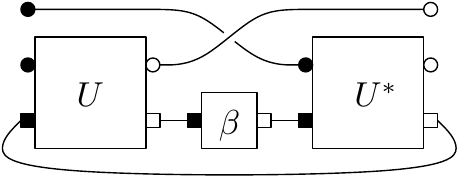} \qquad  \includegraphics{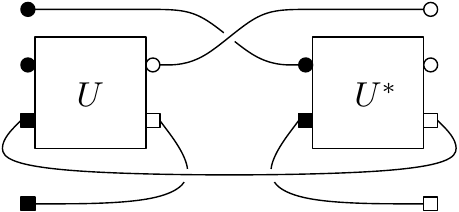} \qquad \includegraphics{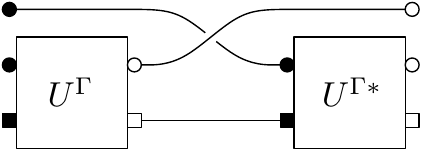} 
\caption{From left to right: the matrix $C_{L_{U,\beta}}^\Gamma$, the Choi matrix of the map $\beta \mapsto C_{L_{U,\beta}}^\Gamma$, and the same Choi matrix, with $U$ replaced by $U^\Gamma$.} 
\label{fig:PPT}
\end{figure}
 
 Since the structure of positive maps between matrix algebras is rather poorly understood, we focus for the moment on a subset of $\mathcal U_{PPT}$, namely
$$ \mathcal U_{PPT} \supseteq \mathcal U_{CPPT} :=  \{U \in \mathcal U_{nk} \, : \, \text{the map } \beta \mapsto C_{L_{U,\beta}}^\Gamma \text{ is \emph{completely} positive}\}.$$

We have the following description of the set $\mathcal U_{CPPT}$, in which, remarkably, the partial transpose of $U$ plays a special role. 

\begin{proposition}
For all $n,k$, we have 
$$ \mathcal U_{CPPT} =  \{U \in \mathcal U_{nk} \, : \, (I_n \otimes U^\Gamma)(F_n \otimes I_k)(I_n \otimes U^\Gamma)^* \geq 0\}.$$
\end{proposition}
\begin{proof}
We use again the fact that complete positivity is characterized by the fact that the Choi matrix is positive semidefinite. In Figure \ref{fig:PPT}, we have depicted in the left image the matrix  $C_{L_{U,\beta}}^\Gamma$, while in the center panel we have the Choi matrix of the map $\beta \mapsto C_{L_{U,\beta}}^\Gamma$. The right-most panel contains the diagram of the same Choi matrix, where we have replaced $U$ by its partial transpose $U^\Gamma$, in order to obtain a nicer expression. The equality of the last two panels contains the proof of the claim. 
\end{proof}

In order to further simplify the description given above, by conjugating the above expression by the pseudo-inverse of the matrix $U^\Gamma$, we are focusing next on the study of the set
\begin{equation}\label{eq:P-cPPT}
\mathcal P_{CPPT} := \{P \in \mathcal M_{nk}^{sa}(\mathbb C) \text{ orthogonal projection} \, : \,  (I_n \otimes P)(F_n \otimes I_k)(I_n \otimes P) \geq 0\},
\end{equation}
and we have $ \mathcal U_{CPPT} =  \{U \in \mathcal U_{nk} \, : \, \mathrm{Ran}(U^\Gamma) \in \mathcal P_{CPPT}\}$.

We have gathered the following properties of the set $\mathcal P_{CPPT}$; we leave the proofs of these simple facts to the reader. 
\begin{enumerate}
\item It is locally unitarily invariant: for all $U \in \mathcal U_n$ and $V \in \mathcal U_k$, if $\tilde P = (U \otimes V)P(U \otimes V)^*$, then
$$ (I_n \otimes \tilde P)(F_n \otimes I_k)(I_n \otimes \tilde P ) =  (U \otimes U \otimes V)(I_n \otimes P)(F_n \otimes I_k)(I_n \otimes P)(U \otimes U \otimes V)^*.$$
\item $I_{nk} \notin \mathcal P_{CPPT}$. In other words, no product unitary lies inside $\mathcal U_{CPPT}$, nor inside $\mathcal U_{PPT}$. As a consequence, we have $\mathcal U_{unital} \cap \mathcal U_{PPT} = \emptyset$.
\item Since $\mathcal U_{const} \subseteq \mathcal U_{PPT}$, if $k=nr$, for any unitary operator $V \in \mathcal U_{nr}$, $P_V \in \mathcal P_{CPPT}$, where $P_V$ is depicted in Figure \ref{fig:P-V}.
\begin{figure}[htbp] 
\includegraphics{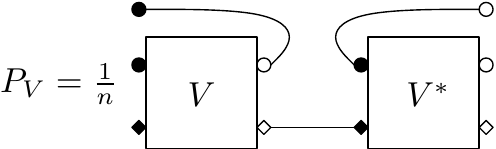} 
\caption{For any unitary operator $V \in \mathcal U_{nr}$, the orthogonal projection $P_V$ depicted here is an element of $\mathcal P_{CPPT}$.} 
\label{fig:P-V}
\end{figure}
\item If $x \otimes y \in \mathrm{Ran} P$, then, for any $x' \in \mathbb C^n$ such that $x' \perp x$, $x' \otimes y \notin \mathrm{Ran} P$.
\item For any $x \in \mathbb C^n$, $\|x\|=1$, and any orthogonal projection $Q \in \mathcal M_k(\mathbb C)$, $xx^* \otimes Q \in \mathcal P_{CPPT}$.
\end{enumerate}

At the level of examples, the only observation here is that $\mathcal U_{const} \subseteq \mathcal U_{PPT}$. We refer the reader to Section \ref{sec:conclusions} for some related open problems.

\section{Bipartite unitary operators producing mixed unitary channels}
\label{sec:mixed}

In this section we investigate the set $\mathcal U_{mixed}$. We provide necessary conditions for a bipartite unitary operator $U$ to belong to $\mathcal U_{mixed}$, and we show that in the case of qubit channels ($n=2$), the sets $\mathcal U_{mixed}$ and $\mathcal U_{unital}$ are equal. 

Recall that the \emph{Kraus operator space} of a quantum channel $L(\rho) = \sum_i E_i \rho E_i^*$ is the space $K(L)=\mathrm{span}\{E_i\}$ \cite{dwi,ydx}; note that $K(L)$ does not depend on the choice of Kraus operators for $L$, since all Kraus representations are related by unitary transformations \cite[Theorem 8.2]{nch}. One of the main observations in \cite{ydx} was that for a mixed unitary channel $L$, $K(L) \cap \mathcal U_n \neq \emptyset$. The next result builds on this remark. 

\begin{proposition}
Let $U \in \mathcal U_{mixed}$ be a bipartite unitary operator. Then, for any unit vector $f \in \mathbb C^k$ and any orthonormal basis $\{e_i\}$ of $\mathbb C^k$, we have
$$\mathrm{span} \left\{ I_n \otimes e_i^* \cdot U \cdot I_n \otimes f \right\}_{i=1}^k \cap \mathcal U_n \neq \emptyset.$$
\end{proposition}
\begin{proof}
For any choice of $f$ and $\{e_i\}$, the operators 
$$E_i := I_n \otimes e_i^* \cdot U \cdot I_n \otimes f $$
are Kraus operators for the channel $L_{U, ff^*}$. Since the channel is mixed unitary, it follows from \cite[Section IV]{ydx} that the linear span of the $E_i$ should contain a unitary operator. 
\end{proof}
\begin{remark}
Note that in the statement above, the set 
$$\mathrm{span} \left\{ I_n \otimes e_i^* \cdot U \cdot I_n \otimes f \right\}_{i=1}^k$$
does not depend on the particular choice of the basis $\{e_i\}$, but only on the vector $f$.
\end{remark}

As a direct consequence of the above result, we obtain the following simple criterion for deciding if a given unitary matrix $U$ is an element of $\mathcal U_{mixed}$.

\begin{corollary}\label{cor:not-mixed}
Let $U \in \mathcal U_{nk}$ be a bipartite unitary operator with the following property:
$$\exists f \in \mathbb C^k \text{ s.t. } \forall e \in \mathbb C^k, \quad I_n \otimes e^* \cdot U \cdot I_n \otimes f \notin \mathcal U_n.$$
Then, $U \notin \mathcal U_{mixed}$.
\end{corollary}

With the help of the criterion above, we present next an example of an element $U \in \mathcal U^B_{block-diag} \setminus \mathcal U_{mixed}$, which shows, in particular, that the inclusion $\mathcal U_{mixed} \subset \mathcal U_{unital}$ is strict; this example is motivated by \cite[Section 4.3]{lst} and \cite[Example 1]{ydx}. Let $U \in \mathcal U_{4\cdot 2}$ be
$$U = e_1e_1^* \otimes I_2 + e_2 e_2^* \otimes \begin{bmatrix} 0 & 1 \\ 1 & 0 \end{bmatrix} + e_3 e_3^* \otimes \frac{1}{\sqrt 2} \begin{bmatrix} 1 & 1 \\ 1 & -1 \end{bmatrix} + e_4 e_4^* \otimes \frac{1}{\sqrt 2} \begin{bmatrix} 1 & i \\ -i & -1 \end{bmatrix}.$$
Obviously, $U \in \mathcal U^B_{cl}$. In the spirit of the criterion above, compute
$$I_4 \otimes \begin{bmatrix} z & w \end{bmatrix} \cdot U \cdot I_4 \otimes \begin{bmatrix} 1 \\ 0 \end{bmatrix} = \mathrm{Diag} \left( z, w, \frac{z + w}{\sqrt 2}, \frac{z - i w}{\sqrt 2}\right).$$
Asking for the diagonal matrix above to be unitary leads to a contradiction, and thus, by Corollary \ref{cor:not-mixed}, we conclude $U  \notin \mathcal U_{mixed}$.

\bigskip

Let us now consider the qubit case $n=2$, which is special because the quantum Birkhoff result holds for qubits \cite{tre, lst,mwo}. 

\begin{proposition}
Let $L:\mathcal M_2(\mathbb C) \to \mathcal M_2(\mathbb C)$ be a completely positive, unital and trace preserving map. Then, there exist  unitary operators $U_1, \ldots, U_k \in \mathcal U_n$ and probabilities $p_i \geq 0$, $\sum_i p_i =1$ such that
$$\forall X \in \mathcal M_2(\mathbb C), \qquad L(X) = \sum_{i=1}^k p_i U_i X U_i^*.$$
\end{proposition}

As a corollary, since every unital channel $L_{U,\beta}$ must be mixed unitary, we obtain the following result. 

\begin{proposition}
For $n=2$ and any $k \geq 2$, we have that
$$\mathcal U_{mixed} = \mathcal U_{unital}.$$
In particular. $U \in \mathcal U_{mixed}$ iff $U \in \mathcal U \cap \mathcal U^\Gamma$.
\end{proposition}

\section{Block-diagonal bipartite unitary operations}
\label{sec:block-diag}

In this section we study the set of block-diagonal operators, $\mathcal U_{block-diag}^{A,B}$. Before proving any results on this class, let us provide another way of writing equation \eqref{eq:def-U-block-diag-A}, which has the benefit of being unique in a certain sense. As a corollary, we deduce that the only unitary transformations which are blcok-diagonal with respect to both sub-systems $A$ and $B$ are given by partial isometries.

\begin{definition}
Two unitary operators $U$ and $V$ are said to be in relation, denoted by $U \sim V$, if there exists a constant $\lambda$ in $\mathbb C$ with $\vert \lambda \vert =1$, such that 
$$U= \lambda\ V\,.$$
\end{definition}

\begin{proposition}\label{prop:uniqueness-decomp-cl}
A bipartite unitary transformation $U \in \mathcal U$ is an element of $\mathcal U_{block-diag}^A$ if and only if it can be written as 
\begin{equation}\label{eq:U-block-diag-Ri}
U = \sum_{i=1}^r U_i \otimes R_i,
\end{equation}
where $U_i$ are unitary operators acting on $\mathbb C^n$ and $R_i$ are partial isometries $R_i : \mathbb C^k \to \mathbb C^k$ such that
\begin{equation}
\sum_{i=1}^r R_iR_i^* = \sum_{i=1}^r R_i^*R_i = I_k
\end{equation}
and $U_i \nsim U_j$ for all $i \neq j$.

Moreover, the decomposition \eqref{eq:U-block-diag-Ri} is unique, up to $\sim$ and permutation of the terms in the sum. 
\end{proposition}
\begin{proof}
Consider two decompositions of a same operator in $\mathcal U_{block-diag}^A$ of the form of \eqref{eq:U-block-diag-Ri} 
\begin{equation}\label{eq:UiVj}
\sum_{i=1}^r U_i \otimes R_i = \sum_{j=1}^s V_j \otimes Q_j,
\end{equation}
with 
\begin{equation}\label{eq:RiQj}
\sum_{i=1}^r R_iR_i^* = \sum_{i=1}^r R_i^*R_i =\sum_{j=1}^s Q_jQ_j^* = \sum_{j=1}^s Q_j^*Q_j = I_k\,,
\end{equation}
$U_p \nsim U_q$ for all $p \neq q$ and $V_l \nsim V_m$ for all $l \neq m$.

For all $i$ in $\{1,\dots,r \}$ and $j$ in $\{1,\dots,s \}$, applying $I \otimes R_i^*$ on the left and $I \otimes Q_j^*Q_j$ on the right, Equation \eqref{eq:UiVj} becomes
$$U_i\otimes R_i^*R_iQ_j^*Q_j= V_j\otimes R_i^*Q_jQ_j^*Q_j =  V_j \otimes R_i^*Q_j\,.$$
This particularly means that
$$R_i^*Q_j=0 \quad \text{ or } \quad \exists \lambda_{ij},  \ \vert \lambda_{ij} \vert =1, \  U_i = \lambda_{ij}\, V_j \quad \text{ and }\quad 
R_i^*R_iQ_j^*Q_j = \dfrac1{\lambda_{ij}}\,  R_i^*Q_j\,.$$
Now note that we have
$$U_i\otimes R_i^*R_i=\sum_{j=1}^sU_i\otimes R_i^*R_iQ_j^*Q_j \neq 0\,.$$
This implies that at least one of the terms in the sum is non-trivial. Moreover, since $V_l \nsim V_m$ for all $l \neq m$, the operator $U_i$ can be in relation with only one of the $V_j$'s. Therefore, we obtain $r=s$ and for all $i$, there exist a unique $j$ such that $U_i = \lambda_{ij}\, V_j $ and $R_i^*R_iQ_j^*Q_j = 1/\lambda_{ij}\,  R_i^*Q_j$. After following the same strategy with $I \otimes Q_jQ_j^*$ on the left and $I \otimes R_i^*$ on the right, we now can deduce that $R_i = 1/\lambda_{ij} Q_j$. The result follows
\end{proof}

Another point of view on block-diagonal unitaries is the fact captured in the next proposition.

\begin{proposition}\label{prop:block-diag-matrix-unitary}
A bipartite unitary operation is block-diagonal if and only if it admits a block-singular value decomposition with respect to $B$:
\begin{equation}
\mathcal U_{block-diag}^A =  \mathcal U_{nk} \cap  \mathcal M_{block-diag}^A,
\end{equation}
where $ \mathcal M_{block-diag}^A$ is the set (see Appendix \ref{sec:block-svd})
$$\mathcal M_{block-diag}^A = \{X = \sum_{i=1}^k X_i \otimes e_if_i^*, \text{ for } X_i \in \mathcal M_n(\mathbb C) \text{ and orthonormal bases } e_i, f_i \text{ of } \mathbb C^k\}.$$
In particular, the sets $\mathcal U_{block-diag}^A$, $\mathcal U_{block-diag}^B$, and $\mathcal U_{block-diag}^A \cap \mathcal U_{block-diag}^B$ are algebraic varieties. 
\end{proposition}
\begin{proof}
For $U \in  \mathcal M_{block-diag}^A$, write
\begin{equation}
UU^* = \sum_{i=1}^k X_iX_i^* \otimes e_ie_i^*.
\end{equation}
The above matrix is the identity if and only if each of its diagonal blocks $X_iX_i^*$ is the identity, and the claim follows. 
\end{proof}

Let us now investigate the relation between the two classes $\mathcal U_{block-diag}^A$ and $\mathcal U_{block-diag}^B$. We start by presenting an algorithm allowing to check if a unitary matrix $U$ in $\mathcal U_{nk} $ belongs to $\mathcal U_{block-diag}^A$. This key result relies on Theorem \ref{thm:block-SVD}.
\begin{proposition}\label{prop:CNS-block-diag-unitary}
Let $U$ be in $\mathcal U_{nk} $. Consider the operators $(X_\alpha)_{\alpha=1,\dots,n^2}$ in $\mathcal M_k(\mathbb C)$ defined, for all $\alpha=1,\dots,n^2$, by
$$X_{\alpha}:=[\mathrm{Tr}\otimes \id]((E_\alpha^*\otimes I)\, U)\,,$$
with $\{E_\alpha\}_{\alpha=1}^{n^2}$ an orthonormal basis of $\mathcal M_n(\mathbb C)$.
Then, $U$ belongs to $\mathcal U_{block-diag}^A$ if and only if the families $\{X_\alpha X_\beta^*\}_{\alpha,\beta=1}^{n^2}$, resp. $\{X_\alpha^* X_\beta\}_{\alpha,\beta=1}^{n^2}$, consist of commuting, normal operators. 
\end{proposition}
\begin{proof}
Fix $U \in \mathcal U_{nk} $. Thanks to Theorem \ref{thm:block-SVD}, the matrix $U$ has a block-diagonal SVD, that is, there exists matrices $U_1, \ldots U_p \in \mathcal M_n(\mathbb C)$ and partial isometries $R_1, \ldots, R_p \in \mathcal M_k(\mathbb C)$ having orthogonal initial, resp.~ final, projections such that
\begin{equation}\label{eq:X-block-SVD}
U = \sum_{i=1}^p U_i \otimes R_i,
\end{equation}
 if and only if the families $\{X_\alpha X_\beta^*\}_{\alpha,\beta=1}^{n^2}$, resp. $\{X_\alpha^* X_\beta\}_{\alpha,\beta=1}^{n^2}$, consist of commuting, normal operators. Then the unitarity of the $(U_i)_{i=1,\dots,p}$'s directly follows from the unitarity of $U$.
\end{proof}

As a first application of the above result let us present the relation between $U_{block-diag}^A$ and $U_{block-diag}^B$ in the case of the qubit space $\mathbb C^2$.

\begin{proposition}
If $n=2$, then
$$\mathcal U_{block-diag}^B \subseteq \mathcal U_{block-diag}^A.$$
\end{proposition}
\begin{proof}
Any element $U\in \mathcal U_{block-diag}^{B}$ can be written as
$$U=e_1f_1^* \otimes U_1+ e_2f_2^* \otimes U_2.$$
In order to apply Proposition \ref{prop:CNS-block-diag-unitary} we consider the orthonormal basis $\{E_{ij}=f_ie_j^*, i,j=1,2\}$. In particular
$$X_{ij}=\delta_{ij}U_j,$$
for all $i,j=1,2$. Consider now the sets 
\begin{align*}
&\{X_{ij}X_{kl}^*, i,j,k,l=1,2\}=\{I,U_1U_2^*,U_2U_1^*\} \text{ and } \\
&\{X_{ij}^*X_{kl}, i,j,k,l=1,2\}=\{I,U_1^*U_2,U_2^*U_1\}.
\end{align*}
It is obvious that these sets consist of commuting, normal operators, finishing the proof.
\end{proof}
\begin{corollary}
In the case $n=k=2$, we have
$$\mathcal U^{A}_{block-diag}  = \mathcal U^{B}_{block-diag}.$$
\end{corollary}

Note however that the inclusion in the above result is strict (in the case $n \geq 3$, $k=2$). For $n=3$, and arbitrary $k \geq 2$, we construct next an example of a unitary operator being in $\mathcal U^{B}_{block-diag}$ but not in $\mathcal U^{A}_{block-diag}$. 

Consider two non commuting unitary operators $V$, $W$ in $\mathcal M_k (\mathbb C)$ and an orthonormal basis $( e_i)$  of $\mathbb C^3$ and define
$$U= e_1e_1^* \otimes I +e_2e_2^* \otimes V+ e_3e_3^* \otimes W.$$
By construction, the operator $U$ belongs to $\mathcal U^{B}_{block-diag}$. Let us now check that it is not in $\mathcal U^{A}_{block-diag}$. Consider the orthonormal basis $\{E_ij=e_ie_j^*, i,j=1,2,3\}$ of $\mathcal M_3(\mathbb C)$, we have for example 
$$X_{11}= I,\quad  X_{22}= V, \quad X_{33}= W.$$
We immediately note that the operators $X_{11}^*X_{22}=V$ and $X_{11}^*X_{33}=W$ do not commute. Since this commutativity is necessary to be in $\mathcal U^{A}_{block-diag}$ (Proposition \ref{prop:CNS-block-diag-unitary}), we conclude that $U$ doesn't belong to $\mathcal U^{A}_{block-diag}$.

Another class of interesting block-diagonal (with respect to the second system, $B$) operators are {circulant} unitary matrices. 

\begin{proposition}
Let $X \in \mathcal M_{nk}(\mathbb C)$ be a circulant matrix. Then $X \in \mathcal M_{block-diag}^B$. In particular, any circulant unitary operator $U$ is block-diagonal with respect to the $B$ factor. 
\end{proposition}
\begin{proof}
In the proof of this result, since we are going to make us of circularity properties, the indices for the matrices we consider are starting at zero. Recall that a matrix $X \in \mathcal M_{nk}(\mathbb C)$ is circulant iff 
$$\forall \, 0 \leq i,j < nk, \qquad X_{ij} = x_{[j-i]_n},$$
where $x \in \mathbb C^{nk}$ is the first row of $X$ and we write $[a]_p = a \mod p$, for any integers $a$ and $p$. Circulant matrices are known to be precisely the matrices which are diagonal in the Fourier basis. Recall that the Fourier matrix (which implements the change of bases between the canonical basis and the Fourier basis) is given by
$$\forall \, 0 \leq i,j < p, \qquad F_p(i,j) = \omega^{ij},$$
where $\omega = \exp(2\pi i /p)$ is a primitive $p$-th root of unity. Finally, since we are going to work with matrices living in a tensor product space, the element $(s,t)$ of the block $(i,j)$ of a matrix $A \in \mathcal M_n(\mathbb C) \otimes \mathcal M_k(\mathbb C)$ is $A_{i\cdot k + s, j \cdot k + t}$.

Now that the notation is fixed, consider a circular unitary matrix $X \in \mathcal M_{nk}(\mathbb C)$ and let $x$ be its first row. For any matrix $A \in \mathcal M_k(\mathbb C)$, define
$$X_A = [\mathrm{id} \otimes \mathrm{Tr}](X \cdot I \otimes A).$$
We show next that the matrices $X_A$ are all circulant, fact which, by Theorem \ref{thm:block-SVD}, suffices to conclude, since all the matrices appearing in the theorem will be simultaneously diagonalizable in the Fourier basis. 

For all $0 \leq i,j < n$, we have
\begin{align*}
X_A(i,j) &= \sum_{s,t = 0}^{k-1} U_{ik+s, jk+t} A_{ts}\\
&= \sum_{s,t = 0}^{k-1} u_{[(j-i)k+(t-s)]_{nk}} A_{ts}.
\end{align*}
The crucial observation is that the above quantity only depends on the difference $j-i$: indeed, if $[(j-i) = (j'-i')]_n$, then there exists some $r$ such that $j'-i' = j-i + nr$, and thus
$$[(j'-i')k+(t-s)]_{nk} = [(j-i)k+(t-s) + nk r]_{nk} = [(j'-i')k+(t-s)]_{nk},$$
showing that the matrix $X_A$ is circular, and finishing the proof. 

The statement about circular unitary operators follows from the general case using Proposition \ref{prop:block-diag-matrix-unitary}.
\end{proof}

We turn next to the study of the unitary operators which are block-diagonal with respect to both systems $A$ and $B$.

\begin{proposition}
A unitary operator $U$ is block diagonal with respect to both tensor factors $A$ and $B$ (i.e.~ $U \in \mathcal U^{A}_{block-diag}\cap \mathcal U^{B}_{block-diag}$) iff 
$$U=\sum_{i=1}^s\sum_{j=1}^r \lambda_{ij} \ Q_i \otimes R_j,$$
where, for all $i=1,\dots,s,\, j=1,\dots,r$, $\vert \lambda_{ij}\vert=1$,
and where $(Q_i)_{i=1,\dots,s}$ , $(R_j)_{j=1,\dots,r}$ are two family of orthogonal partial isometries respectively on $\mathbb C ^n$ and $\mathbb C ^k$ satisfying
\begin{equation}\label{eq:QiRj}
\sum_{i=1}^s Q_iQ_i^* = \sum_{i=1}^s Q_i^*Q_i = I_n \quad \text{, }\quad \sum_{j=1}^r R_jR_j^* = \sum_{j=1}^r R_j^*R_j =I_k \, .
\end{equation}
\end{proposition}
\begin{proof}
Let $U$ be an element in the intersection $\mathcal U_{block-diag}^A \cap \mathcal U_{block-diag}^B $. Then, $U$ admits both decompositions 
\begin{equation*}
U = \sum_{i=1}^s Q_i \otimes V_i= \sum_{j=1}^r U_j \otimes R_j\,.
\end{equation*}
Applying $Q_i^* \otimes R_j^*$ on the left, we obtain
$$ Q_i^*Q_i \otimes R_j^*V_i =Q_i^* U_j \otimes R_j^* R_j\,.$$
Then there exists $\mu_{ij}\neq 0$ such that 
$$Q_i^*Q_i = \mu_{ij}  Q_i^* U_j  \qquad \text{ and }\qquad R_j^*V_i = \dfrac1{\mu_{ij} }\,R_j^* R_j\,.$$
Now since the $Q_i$'s and the $R_j$'s satisfy \eqref{eq:QiRj} we end up with
$$ U_j = \sum_{i=1}^s \overline{\mu_{ij} } \ Q_i \qquad \text{ and }\qquad V_i =  \sum_{j=1}^r \dfrac1{\overline{\mu_{ij} }}\ R_j\,.$$
Since the operators $U_j$ and $V_i$ are unitary, we conclude that $\vert \mu_{ij} \vert=1$ and that gives the result.
\end{proof}

Finally, we compute next the (real) dimension of $\mathcal U_{block-diag}^A$ and $\mathcal U_{block-diag}^A \cap \mathcal U_{block-diag}^B$.
\begin{proposition}
The real dimension of the algebraic variety $\mathcal U_{block-diag}^A$ is 
$$\dim \mathcal U_{block-diag}^A = \begin{cases}
k^2 & \qquad \text{ if } n=1\\
k(n^2+2k-2) & \qquad \text{ if } n>1.
\end{cases}$$

The real dimension of the algebraic variety $\mathcal U_{block-diag}^A \cap \mathcal U_{block-diag}^B$
$$\dim (\mathcal U_{block-diag}^A \cap \mathcal U_{block-diag}^B)= \begin{cases}
(nk)^2& \qquad \text{ if } \min(n,k)=1\\
 2n^2+ 2k^2+nk -2n-2k& \qquad \text{ if } \min(n,k)>1.
\end{cases}
$$
\end{proposition}
\begin{proof}

 Let us first perform a heuristic parameter counting for a generic element
$$\mathcal U_{block-diag}^A \ni U = \sum_{i=1}^k U_i \otimes e_i f_i^*.$$
The choice of the two orthonormal bases $\{e_i\}$ and $\{f_i\}$ in \eqref{eq:X-block-SVD} corresponds to a total of $2k^2$ real parameters, since $\dim_{\mathbb R} \mathcal U_k = k^2$. Each matrix $U_i$ accounts for $n^2$ real parameters, so, in total, we get $kn^2$ extra real parameters. However, in each term $U_i \otimes e_if_i^*$, two of the three complex phases of $X_i, e_i, f_i$ are redundant, so we have over counted $2k$ real parameters. We conclude that the real dimension of $\mathcal U_{block-diag}^A$ should be $2k^2+kn^2-2k$. The above reasoning over-counts in the case $n=1$. Indeed, in that case one can ignore the matrices $U_i$ (the phase can be absorbed in the $e_if_i^*$ part), and we are left with an unitary matrix, which counts for $k^2$ real parameters. The rigourous proof of this result is very similar to the one of Proposition \ref{prop:dim-M-block-diag}, and is left to the reader. 

Let us now find the dimension of the intersection. Similarly, let us count parameters for a generic element of the form
$$
\mathcal U_{block-diag}^A \cap \mathcal U_{block-diag}^B \ni U=\sum_{i=1}^n\sum_{j=1}^k \lambda_{ij} e_if_i^* \otimes g_jh_j^*
$$
where, for all $i=1,\dots,n,\, j=1,\dots,k$, $\vert \lambda_{ij}\vert=1$, and $(e_i)$, $(f_i)$, $(g_j)$ and $(h_j)$ are orthonormal bases of $\mathbb C^n$ and $\mathbb C^k$, respectively.

The choice of the four orthonormal bases corresponds to a total of $2n^2+2k^2$ real parameters, the choice of the coefficients $nk=2nk-nk$. Since, in  $\lambda_{ij} e_if_i^* \otimes g_jh_j^*$, all the phases can be absorbed in the coefficient $\lambda_{ij}$, we have over counted $2n+2k$ real parameters. Again, the case $\min(n,k)=1$ is degenerated, since any unitary operator is of the desired form. 
\end{proof}

\section{Further relations between unitary classes}\
\label{sec:relations}

As discussed in the introduction, the following chain of inclusions holds:
\begin{equation}\label{eq:chain-inclusions}
\mathcal U_{block-diag}^A \subseteq \mathcal U_{probl-lin} \subseteq \mathcal U_{prob} \subseteq \mathcal U_{mixed} \subseteq \mathcal U_{unital}.
\end{equation}

We discuss in this section the situations when some of the above inclusions are equalities. See Section \ref{sec:conclusions} for some related open questions. 

\begin{theorem}\label{thm:block-diag-prob-lin}
The sets $\mathcal U_{block-diag}^A$ and $\mathcal U_{prob-lin}$ are equal.
\end{theorem}
\begin{proof}
Without loss of generality, we assume that the unitary operators $U_i$ are different up to a phase, i.e. $U_i \not\sim U_j$, for all $i \neq j$.

Due to the decomposition of any matrix into a linear combination of at most four matrices in $\mathcal M_k^{1,+}(\mathbb C)$ (positive and negative hermitian parts and their equivalents for the anti-hermitian part), the $p_i$'s can be extended by linearity to positive functionals on $\mathcal M_k(\mathbb C)$. By Riesz Theorem, for each $i$, there exists a matrix $M_i$ such that
$$\forall X \in \mathcal M_k(\mathbb C), \quad p_i(X)= \mathrm{Tr}(M_i X)\,.$$
Note now that the values $p_i(\beta)= \mathrm{Tr}(M_i \beta)$ are non-negative for all $\beta$ in $\mathcal M_k^{1,+}(\mathbb C)$. Therefore the matrix $M_i$ is actually positive semi-definite since we have
$$\mathrm{Tr}(M_i \beta)= p_i(\beta)=\overline{p_i(\beta)}= \overline {\mathrm{Tr}(M_i \beta)}= \mathrm{Tr}(M_i^* \beta)\,, \quad \forall \beta \in \mathcal M_k^{1,+}(\mathbb C) .$$ 
We then conclude that the matrices $M_i$ can be written as $M_i= R_i R_i^*=R_i^*R_i=R_i^2$ for some hermitian and positive semi-definite $R_i$.

Moreover, using that for all $\beta$ in $\mathcal M_k^{1,+}(\mathbb C)$
$$\sum_{i=1}^r p_i(\beta)= \mathrm{Tr}\Big( \Big[\sum_{i=1}^r M_i \Big] \beta\Big)=1,$$
it follows that $\sum_{i=1}^r M_i=I$.
Then, for all $\beta$ in $\mathcal M_k^{1,+}(\mathbb C)$, we can write the quantum channel $L_{U,\beta}$ as
$$L_{U,\beta}(X) = \mathrm{Tr}_B(U\cdot X \otimes \beta \cdot U^*)=\mathrm{Tr}_B\Big( \sum_{i=1}^r U_i \otimes R_i\cdot X \otimes \beta \cdot (U_i \otimes R_i)^*\Big).$$
By linearity, the previous equality gives
$$\forall Y \in \mathcal M_{nk}(\mathbb C), \quad \mathrm{Tr}_B(U\cdot Y \cdot U^*)=\mathrm{Tr}_B\Big( \sum_{i=1}^r U_i \otimes R_i\cdot Y \cdot (U_i \otimes R_i)^*\Big).$$
Then, by definition of the partial trace, we obtain the following equivalences
\begin{align}
\nonumber \forall Y \in \mathcal M_{nk}(\mathbb C), A \in \mathcal M_{n}(\mathbb C), \quad& \mathrm{Tr}(U\cdot Y \cdot U^* A\otimes I)=\mathrm{Tr}\Big( \sum_{i=1}^r U_i \otimes R_i\cdot Y \cdot (U_i \otimes R_i)^* A \otimes I\Big) \\
\nonumber\forall Y \in \mathcal M_{nk}(\mathbb C), A \in \mathcal M_{n}(\mathbb C), \quad& \mathrm{Tr}(Y \cdot U^* A\otimes IU)=\mathrm{Tr}\Big(Y \cdot \sum_{i=1}^r   U_i ^*AU_i\otimes M_i\Big)\\
\forall A \in \mathcal M_{n}(\mathbb C), \quad&U^* A\otimes IU=\sum_{i=1}^r   U_i ^*AU_i\otimes M_i\label{eq:AtenseurI}.
\end{align}
Let us now study \eqref{eq:AtenseurI}. Consider two matrices $A,B$ in $\mathcal M_{n}(\mathbb C)$. Since $U^* AB\otimes IU=U^* A\otimes IU\cdot U^* B\otimes IU$, we obtain from \eqref{eq:AtenseurI}
\begin{equation}\label{eq:ABtenseurI}
\sum_{i=1}^r   U_i ^*ABU_i\otimes M_i = \sum_{i,j=1}^r   U_i ^*AU_iU_j ^*BU_j\otimes M_iM_j.
\end{equation}
Applying \eqref{eq:ABtenseurI} with $A=ee^*$ and $B=ff^*$ where $e$ and $f$ are orthogonal vectors of $\mathbb C^n$, it directly follows
$$0= \sum_{i\neq j}  U_i ^* ee^* U_i  U_j ^*ff^*U_j \otimes M_iM_j.$$
Taking the trace, we obtain
$$\forall e \perp f, \quad 0= \sum_{i\neq j}   \vert \langle U_i ^* e, U_j ^*f\rangle \vert^2\ \mathrm{Tr}( M_iM_j)= \sum_{i\neq j}   \vert \langle  e,U_i  U_j ^*f\rangle \vert^2 \ \mathrm{Tr}( M_iM_j).$$
Note that the previous equation is actually a sum of non-negative terms equals to $0$. Therefore, we conclude that for $i\neq j$
\begin{equation}\label{eq:M_iM_j=0}
\mathrm{Tr}( M_iM_j)=0\quad \text{ or } \quad  \forall e \perp f, \quad \langle  e,U_i  U_j ^*f\rangle =0.
\end{equation}
Now since for all $i \neq j$, $U_i \not\sim U_j$, the unitary matrix $U_i  U_j ^*$ is not a multiple of the identity. Thus, we claim that we can find orthogonal vectors $e$ and $f$, such that $ \langle  e,U_i  U_j ^*f\rangle \neq0$. Indeed, the matrix $U_i  U_j ^*$ is diagonalizable in an orthonormal basis $(u_p)_{p=1,...,n}$ with related eigenvalues $(\mu_p)_{p=1,...,n}$. Since $U_i  U_j ^*\not\in \mathbb C I$, all the eigenvalues are not equal, for instance $\mu_1 \neq \mu_2$. Let us consider the orthogonal vectors $e= u_1 -u_2$ and $f=u_1+u_2$. We can easily check that $\langle  e,U_i  U_j ^*f\rangle=\mu_1-\mu_2\neq 0$.

Finally, we deduce from \eqref{eq:M_iM_j=0} that  $\mathrm{Tr}( M_iM_j)=0$ for all $i\neq j$ and thus $M_i M_j=0$. The $M_i$'s being orthogonal and sum to the identity, we define the matrix $V= \sum_{i=1}^r U_i \otimes R_i$ belonging to $\mathcal U^{A}_{cl}$. We can now easily check that the unitary operators $U$ and $V$ induce the same quantum channels for all $\beta$ in $\mathcal M_k^{1,+}(\mathbb C)$. By Lemma \ref{lem:same-channels}, there exists $W$ of $\mathcal U_k$ such that $U= (I \otimes W) V$; this proves the result.
\end{proof}

\begin{proposition}\label{prop:block-diag-unital-qubit}
If $n=2$, then $\mathcal U_{block-diag}^A = \mathcal U_{unital}$. In particular, the chain of inclusions \eqref{eq:chain-inclusions} collapses:
$$\mathcal U_{block-diag}^A = \mathcal U_{probl-lin} = \mathcal U_{prob} = \mathcal U_{mixed} = \mathcal U_{unital}.$$
\end{proposition}
\begin{proof}
Consider $\{e_1,e_2\}$ the canonical orthonormal basis of $\mathbb C^2$ and consider a matrix $U$ in $ \mathcal U_{unital}$ written as $2\times 2$ block matrices 
$$U=\left(\begin{array}{cc} A&B\\C&D
\end{array}
\right)=e_1e_1^*\otimes A+e_1e_2^*\otimes B+e_2e_1^*\otimes C+e_2e_2^*\otimes D$$
with $A,B,C$ and $D$ in $\mathcal M_{k}(\mathbb C)$. As proved in Theorem \ref{unital}, both $U$ and $(U^\Gamma)^T$ are unitary matrices and therefore
\begin{equation}\label{system_unital}
\begin{array}{cccccccccc}
AA^*+ BB^*&=&A^*A+ B^*B&=&I,& 
AA^*+ CC^*&=&A^*A+ C^*C&=& I,\\
CC^*+ DD^*&=&C^*C+ D^*D&=&I,&
DD^*+ BB^*&=&D^*D+ B^*B&=&I,\\
AC^*+BD^*&=&A^*C+B^*D&=&0, & 
AB^*+CD^*&=&A^*B+C^*D&=&0. 
\end{array}
\end{equation}
Our aim is to applied Proposition \ref{prop:CNS-block-diag-unitary}. Consider the orthonormal basis $\{E_{ij}=e_je_j^*,i,j=1,2\}$ of $\mathcal M_2(\mathbb C)$, it is clear that 
\begin{align*}
\{X_{ij}X_{kl}^*,i,j,k,l=1,2\}&=\lbrace XY^*\, \vert \,X,Y=A,B,C,D\rbrace\\\{X_{ij}^*X_{kl},i,j,k,l=1,2\}&=\lbrace X^*Y\, \vert \,X,Y=A,B,C,D\rbrace.
\end{align*} 
Let us now prove that $\lbrace XY^*\, \vert \,X,Y=A,B,C,D\rbrace$ and $\lbrace X^*Y\, \vert \,X,Y=A,B,C,D\rbrace$ consist of normal commuting matrices. It can be noticed that it is sufficient to check that, for all $X,Y,Z=A,B,C,D$, 
$$XY^*Z=ZY^*X\,.$$
Using the symmetry $A,D$ in \eqref{system_unital} together with the symmetry $B,C$, the 64 different cases boil down to 11 non-trivial cases : $AA^*B, AA^*D, AB^*B, AB^*C, AB^*D, AD^*B, AD^*D, BA^*C, BA^*D, BB^*C$ and $BC^*D$. Each of them can be easily checked as for instance
\begin{align*}
AA^*B=(I-BB^*)B= B (I-B^*B)=BA^*A
\end{align*}
or
\begin{align*}
BC^*D= B(C^*D)= B(-A^*B)= (-BA^*)B= DC^*B\,.
\end{align*}
The result then holds.
\end{proof}

\begin{remark}
A similar result has been obtained in \cite[Theorem 9]{kmwy}, under more stringent assumptions. More precisely, it is shown in \cite{kmwy} that, when $n=2$, $\mathcal U_{block-diag}^A = \mathcal U_{prob}$, assuming that the unitary operators appearing in the mixed-unitary decomposition of channels are linearly independent.
\end{remark}

Swapping the roles of $n$ and $k$, we obtain the following result. \begin{proposition}
If $k=2$, then
$\mathcal U_{block-diag}^B= \mathcal U_{unital}\,.$
\end{proposition}

\section{Conclusions and open questions}
\label{sec:conclusions}

We end this work with a list of questions that we have left unanswered (or even untouched). We hope to get back to some of these problems in some future work. 

We start with the problem of computing the dimension of the algebraic variety $\mathcal U_{unital}$; recall that previously, we have looked at the enveloping tangent space of this variety, at some particular points. 

\begin{conjecture}\label{cnj:dim-U-unital}
Show that $\dim \mathcal U_{unital} = kn^2+nk^2-nk$.
\end{conjecture}

It has been showed in Theorem \ref{thm:block-diag-prob-lin} that any operator in the set $\mathcal U_{prob-lin}$ (which is a subset of $\mathcal U_{mixed}$) is block diagonal, with respect to the system $A$. Moreover, in the qubit case $n=2$, we have $\mathcal U_{block-diag}^A = \mathcal U_{mixed}$, see Proposition \ref{prop:block-diag-unital-qubit}. We conjecture that this equality always hold, and that the technical restrictions appearing in the definition of $\mathcal U_{prob-lin}$ are actually superfluous. 

\begin{conjecture}
For all values of $n,k$, it holds that
$$\mathcal U_{block-diag}^A = \mathcal U_{mixed}.$$
\end{conjecture}

Regarding bipartite unitary operators producing PPT channels, we have left the following problem open. 
 
 \begin{question}
 Given a subspace $V \subseteq \mathbb C^n \otimes \mathbb C^k$, let $P \in \mathcal M_{nk}^{sa}(\mathbb C)$ be the orthogonal projection on $V$. Characterize the set of subspaces $V$ such that 
 $$\mathcal M_{n^2k}^{sa}(\mathbb C) \ni (I_n \otimes P)(F_n \otimes I_k)(I_n \otimes P) \geq 0,$$
 where $F$ is the flip operator. 

 \end{question}
 
 This brings us to the problem of characterizing the set $\mathcal U_{EB}$ and comparing it to $\mathcal U_{PPT}$ (at the level of quantum states, this would be the fact that the PPT criterion for separability is necessary in all dimensions, and sufficient for $nk\leq 6$). 
 
\begin{question}
Provide a description of the set $\mathcal U_{EB}$. For which values of $n,k$, is it true that $\mathcal U_{PPT} = \mathcal U_{EB}$?
\end{question}

At the level of examples, beside the obvious inclusion $\mathcal U_{const} \subseteq \mathcal U_{EB}$, we also have\footnote{We thank Siddharth Karumanchi for pointing this out to us.}, when $n=k$, 
$$\mathcal U^A_{block-diag} \cdot F_n \subseteq \mathcal U_{EB}.$$
Indeed, for a unitary operator $U = \left(\sum_{i=1}^n U_i \otimes e_i f_i^* \right) \cdot F_n$, the corresponding quantum channel reads
$$L_{U, \beta}(\rho) = \sum_{i=1}^n \langle f_i  , \rho f_i \rangle \,  U_i \beta U_i^*,$$
which is entanglement-breaking.

Finally, we consider the following classe of bipartite unitary matrices yielding channels of interest in quantum information theory. 
\begin{align*}
\mathcal U_{CQ} &= \{U \in \mathcal U \, | \, \forall \beta \in \mathcal M_k^{1,+}(\mathbb C), L_{U,\beta} \text{ is a classical-quantum channel}\}\\
\mathcal U_{QC} &= \{U \in \mathcal U \, | \, \forall \beta \in \mathcal M_k^{1,+}(\mathbb C), L_{U,\beta} \text{ is a quantum-classical  channel}\}\\
\mathcal U_{CC} &= \{U \in \mathcal U \, | \, \forall \beta \in \mathcal M_k^{1,+}(\mathbb C), L_{U,\beta} \text{ is a classical-classical  channel}\}
\end{align*}

The study of these classes has been initiated in \cite{kmwy,kmwy2}, where mainly the qubit case $n=2$ has been discussed. The structure of these operators in the general case remains open.

\begin{question}
Characterize the sets $\mathcal U_{CQ}$, $\mathcal U_{QC}$, and $\mathcal U_{CC}$.
\end{question}

\appendix

\section{Necessary and sufficient conditions for the existence of a block-SVD}\label{sec:block-svd}

In this section, we establish necessary and sufficient conditions for the existence of a block singular value decomposition of a bipartite operator with respect to the second sub-system $B$. Moreover, we present an algorithm for obtaining such a decomposition when it does exist. These results are inspired from \cite{ash,dvb}, see also \cite[Theorem (2.5.5) and Section 7.3, Problem 25]{hjo}.

We shall denote, for any matrices $X \in \mathcal M_n(\mathbb C) \otimes \mathcal M_k(\mathbb C)$ and $A \in \mathcal M_n(\mathbb C)$,
$$X_A := [\mathrm{Tr} \otimes \mathrm{id}](X \cdot A \otimes I_k).$$

\begin{theorem}\label{thm:block-SVD}
Let $X \in \mathcal M_n(\mathbb C) \otimes \mathcal M_k(\mathbb C)$. The following conditions are equivalent:
\begin{enumerate}
\item The matrix $X$ has a \emph{block-diagonal SVD}: there exists matrices $X_1, \ldots X_p \in \mathcal M_n(\mathbb C)$ and partial isometries $R_1, \ldots, R_p \in \mathcal M_k(\mathbb C)$ having orthogonal initial, resp.~ final, projections such that
$$X = \sum_{i=1}^p X_i \otimes R_i.$$
\item The families $\{X_A X_B^*\}_{A,B \in \mathcal M_n(\mathbb C)}$, resp.~ $\{X_A^* X_B\}_{A,B \in \mathcal M_n(\mathbb C)}$, consist of commuting, normal operators. 
\item The families $\{X_\alpha X_\beta^*\}_{\alpha,\beta=1}^{n^2}$, resp.~ $\{X_\alpha^* X_\beta\}_{\alpha,\beta=1}^{n^2}$, consist of commuting, normal operators. We write $X_\alpha := X_{E_\alpha^*}$ for an orthonormal basis $\{E_\alpha\}_{\alpha=1}^{n^2}$ of $\mathcal M_n(\mathbb C)$, in such a way that $X = \sum_{\alpha=1}^{n^2} E_\alpha \otimes X_\alpha$.
\end{enumerate}
We denote the set of matrices $X$ satisfying the above condition(s) by $\mathcal M_{block-diag}^A$
$$\mathcal M_{block-diag}^A := \{X \in \mathcal M_{nk}(\mathbb C) \, | \, X = \sum_{i=1}^k X_i \otimes e_if_i^*\}.$$
The set $\mathcal M_{block-diag}^B$ is defined in a similar way, by swapping the roles of the two tensor factors. 
\end{theorem}
\begin{proof}
The implication $(2) \implies (3)$ is obvious. Let us first show $(1) \implies (2)$. For a matrix $X$ as in \eqref{eq:X-block-SVD}, we have
$$X_A X_B^* = \sum_{i=1}^p \mathrm{Tr}(X_iA) \overline{\mathrm{Tr}(X_iB)} P_i,$$
where $P_i = R_iR_i^*$ is the orthogonal projection on the image of the partial isometry $R_i$. This relation shows that the matrices $\{X_A X_B^*\}_{A,B \in \mathcal M_n(\mathbb C)}$ are normal and diagonalizable in the same orthonormal basis of $\mathbb C^n$. A similar argument shows the result for the matrices $\{X_A^* X_B\}_{A,B \in \mathcal M_n(\mathbb C)}$

Let us now show $(3) \implies (1)$. The fact that the normal matrices $X_\alpha X_\beta^*$ commute implies they have the same set of eigenprojectors $P_i$:
\begin{equation}\label{eq:alpha-beta-star}
\forall \alpha,\beta, \quad X_\alpha X_\beta^* = \sum_{i=1}^p \lambda_i^{(\alpha,\beta)} P_i.
\end{equation}
In the same vein, we have, for another set of orthogonal eigenprojectors $Q_i$:
\begin{equation}\label{eq:alpha-star-beta}
\forall \alpha,\beta, \quad X_\alpha^* X_\beta = \sum_{i=1}^{p'} \mu_i^{(\alpha,\beta)} Q_i.
\end{equation}
Letting $\alpha=\beta$ in \eqref{eq:alpha-beta-star} and \eqref{eq:alpha-star-beta} and using the fact that the matrices $X_\alpha X_\alpha^*$ and $X_\alpha^* X_\alpha$ have the same (positive) eigenvalues (counting multiplicities), we have that $p=p'$ and there exists permutations $\sigma_\alpha \in \mathcal S_p$ and complex numbers $\lambda_i^{(\alpha)}$ such that
$$X_\alpha = \sum_{i=1}^r \lambda_i^{(\alpha)} R_i^{(\alpha)}$$
for some partial isometries $R_i^{(\alpha)}$ having initial projection $Q_{\sigma_\alpha(i)}$ and final projection $P_i$. Plugging the last expression into \eqref{eq:alpha-beta-star} and \eqref{eq:alpha-star-beta}, we find that the permutations $\sigma_\alpha$ must be equal; we shall assume, by re-ordering the eigenprojectors $Q_i$, that these permutations are all equal to the identity. Using similar arguments, the partial isometries $R_i^{(\alpha)}$ cannot depend on $\alpha$, and we write $R_i^{(\alpha)} = R_i$. We have thus 
$$X_\alpha = \sum_{i=1}^r \lambda_i^{(\alpha)} R_i,$$
and thus
\begin{align*}
X &= \sum_{\alpha=1}^{n^2} E_\alpha \otimes \left( \sum_{i=1}^r \lambda_i^{(\alpha)} R_i \right) \\
&=\sum_{i=1}^r \left( \sum_{\alpha=1}^{n^2} \lambda_i^{(\alpha)} E_\alpha \right) \otimes R_i.
\end{align*}
Setting
$$X_i := \sum_{\alpha=1}^{n^2} \lambda_i^{(\alpha)} E_\alpha$$
concludes the proof.
\end{proof}

\begin{remark}
The above result provides us with an efficient way of checking whether a given bipartite operator $X$ has a block-SVD: pick a basis of $\mathcal M_n(\mathbb C)$ (e.g.~the usual matrix units) and check the condition in (3) above. 
\end{remark}

\begin{remark}
Restricting to  $\alpha = \beta$ in the condition (3) does not yield an equivalent statement, as it is shown in the example below. For $X \in \mathcal M_{2 \cdot 2}(\mathbb C)$ given by
\begin{equation}
X = \begin{bmatrix}
1 & 0 & 0 & 1 \\
0 & 2 & 1 & 0 \\
0 & 1 & 1 & 0 \\
1 & 0 & 0 & 1 
\end{bmatrix}
\end{equation}
with the choice of the canonical matrix units for $\mathcal M_2(\mathbb C)$, the matrices $X_\alpha$ are $X_{11} = \begin{bmatrix} 1 & 0 \\ 0 & 2 \end{bmatrix}$, $X_{12}=X_{21} = \begin{bmatrix} 0 & 1 \\ 1 & 0 \end{bmatrix}$, and $X_{22} = I_2$. All the matrices $X_\alpha X_\alpha^*$ are diagonal, hence they are normal and commute. However, the matrices $X_{11}X_{22}^*$ and $X_{12}X_{22}^*$ do not commute, so the $4 \times 4$ matrix $X$ does not satisfy the equivalent conditions from Theorem \ref{thm:block-SVD}, hence it does not have a block-SVD.
\end{remark}

\begin{corollary}
The set $\mathcal M_{block-diag}^A$ is a real algebraic variety.
\end{corollary}
\begin{proof}
A matrix $X$ belongs to $\mathcal M_{block-diag}^A$ if and only if the two families of $n^4$ matrices $\{X_\alpha X_\beta^*\}$ and $\{X_\alpha^* X_\beta\}$ commute; these commutations conditions can be restated as (degree 4) polynomial conditions in the real and imaginary parts of the elements of $X$. 
\end{proof}

Before computing in the next proposition the dimension of the real algebraic variety $\mathcal M_{block-diag}^A$, we would like to give a heuristic argument in the form of parameter counting. The choice of the two orthonormal bases $\{e_i\}$ and $\{f_i\}$ in \eqref{eq:X-block-SVD} corresponds to a total of $2k^2$ real parameters, since $\dim_{\mathbb R} \mathcal U_k = k^2$. Each matrix $X_i$ accounts for $2n^2$ real parameters, so, in total, we get $2kn^2$ extra real parameters. However, in each term $X_i \otimes e_if_i^*$, two of the three complex phases of $X_i, e_i, f_i$ are redundant, so we have over counted $2k^2$ real parameters. We conclude that the real dimension of $\mathcal M_{block-diag}^A$ should be $2k^2+2kn^2-2k$, fact which we rigorously prove next. 

\begin{proposition}\label{prop:dim-M-block-diag}
The real dimension of the algebraic variety $\mathcal M_{block-diag}^A $ is $2k(n^2+k-1)$.
\end{proposition}

\begin{proof}
For the terminology and the results used in this proof, we refer the reader to \cite[Chapter 11]{har}. Let us introduce the flag manifold (see \cite[Example 8.34]{har} or \cite[Section 4.9]{bzy})
$$\mathrm{Fl}_k = \mathcal U_k / (\mathcal U_1^k),$$
which has real dimension $\dim_{\mathbb R} \mathrm{Fl}_k = k^2-k$. Consider the map
\begin{align*}
\varphi :\mathcal M_n(\mathbb C)^k \times \mathrm{Fl}_k \times \mathrm{Fl}_k  &\to  \mathcal M_{block-diag}^A\subset\mathcal M_n(\mathbb C)\otimes\mathcal M_k(\mathbb C) \\
\left((A_i)_{i=1}^k,([e_i])_{i=1}^k,([f_i])_{i=1}^k \right)& \mapsto  \sum_{i=1}^kA_i \otimes e_if_i^*.
\end{align*}
Obviously, $\varphi$ is surjective, so we have $\dim_{\mathbb R} \mathcal M_{block-diag}^A \leq 2k(n^2+k-1)$. To get the reverse inequality, define $\widetilde{\mathcal M_n(\mathbb C)^k}$ to be the set of pairwise distinct $k$-tuples of matrices. It is trivial to check that the map $\tilde \varphi$, obtained by restricting $\varphi$ to $\widetilde{\mathcal M_n(\mathbb C)^k}$, is $k!$-to-one, so the conclusion follows.
\end{proof}

\end{document}